\documentclass{article}
\usepackage{amsfonts}
\usepackage{amsmath}
\usepackage{amssymb}
\usepackage{amsthm}

\newtheorem{theorem}{Theorem}

\newtheorem{corollary}[theorem]{Corollary}

\newtheorem{definition}[theorem]{Definition}

\newtheorem{lemma}[theorem]{Lemma}

\newtheorem{proposition}[theorem]{Proposition}
\newtheorem{remark}[theorem]{Remark}

\def\canBasis{\varepsilon}
\def\idemProd{\Gamma}

\title {\textsc{Multicomplex Ideals, Modules and Hilbert Spaces}}
\author{Derek Courchesne$^1$ $\quad$ and $\quad$ Sébastien Tremblay$^2$}
\date{December 5, 2024}

\newcommand{\ket}[1]{| #1 \rangle}
\newcommand{\bra}[1]{\langle #1 |}
\newcommand{\braket}[2]{\langle #1 | #2 \rangle}

\begin{document}


\maketitle

\begin{abstract}
In this article we study some algebraic aspects of multicomplex numbers $\mathbb M_n$. For $n\geq 2$ a canonical representation is defined in terms of the multiplication of $n-1$ idempotent elements. This representation facilitates computations in this algebra and makes it possible to introduce a generalized conjugacy $\Lambda_n$, i.e. a composition of the $n$ multicomplex conjugates $\Lambda_n:=\dagger_1\cdots \dagger_n$, as well as a multicomplex norm. The ideals of the ring of multicomplex numbers are then studied in details, free $\mathbb M_n$-modules and their linear operators are considered and, finally, we develop Hilbert spaces on the multicomplex algebra.
\end{abstract}


\section{Introduction}
The pioneering work of Corrado Segre in 1892 introduced an infinite set of algebras, termed bicomplex numbers $\mathbb M_2$, tricomplex numbers $\mathbb M_3$, and so forth, extending the realm of complex numbers $\mathbb C \sim \mathbb M_1$ to higher dimensions. About a century later, Price wrote his book {\em An introduction to Multicomplex Spaces and Functions} \cite{GBPrice}, the first comprehensive work on the subjet.  More recently Luna-Elizarrara's {\em et al.} published a comprehensive study \cite{LunaElizarrarásShapiroStruppaVajiac2015} of the analysis and geometry of bicomplex numbers.

Bicomplex numbers and their extensions, the multicomplex numbers, have now seen many applications across Science. In mathematical analysis, many fundamental results were extended to the bicomplex setting: see \cite{Elgourari2021, Hammam2024, Kravchenko_2008, Kumar2022, FiniteHilbert} for extensions of well-known results from operator theory and functional analysis; see \cite{Campos_Kravchenko_2012, Campos_Kravchenko_Méndez_2012, LunaElizarrarásShapiroStruppaVajiac2015, LunaElizarrarásShapiroStruppaVajiac2013} for extensions of complex analysis to the bicomplex numbers; and \cite{StruppaVajiac2014} for the extension of differential equations to the multicomplex settings. In physics, the Schr\"odinger equation and the theory of quantum physics and general relativity were extended to the bicomplex or hyperbolic settings \cite{BCHarmonic, BCCoulomb, BCQM1, BCQM2} and the hyperbolic numbers \cite{Schafer2014, Ulrych2005, Ulrych2010}, a subalgebra of the bicomplex numbers. Following the work \cite{BCQM1} of D. Rochon and S. Tremblay on the bicomplex Sch\"odinger equation, the authors of \cite{Millwater2014} generalized the linear and nonlinear Schr\"odinger equation to the multicomplex settings.

The work \cite{BCQM1, BCQM2} of D. Rochon and S. Tremblay relies heavily on algebraic properties of the bicomplex numbers that were studied in details in \cite{BCHyper}. These algebraic properties were not as extensively studied in the multicomplex settings as they were in \cite{BCHyper}, even in the previous cited paper \cite{Millwater2014} on the multicomplex linear and nonlinear Schr\"odinger. The goal of this paper is to fill in the gaps in the algebraic properties of the multicomplex numbers.

Our objectives are three-fold. Firstly, we investigate the abstract algebraic structure of the multicomplex numbers. More precisely, the ideals and the quotient rings of the multicomplex numbers will be studied as well as the $\mathbb M_n$-modules. Secondly, we introduce a canonical idempotent representation written in terms of the $n$ multicomplex conjugates and show how this representation is more natural in the treatment of the algebraic operations with the multicomplex numbers. At the same time, we use this idempotent representation
to introduce a useful ring-valued norm on the multicomplex numbers. Thirdly, we introduce the definition of a multicomplex Hilbert space to provide the mathematical background for the different applications of the theory, such as in mathematical analysis, physics, and applied sciences.

\section{Multicomplex numbers algebra}
\label{section : multicomplexDef}

\subsection{Definition}

The $n$th multicomplex numbers space $\mathbb{M}_n$ is defined by the space obtained after a total of $n \in \mathbb{N}$ successive complexifications of the reals, each time introducing a new imaginary unit $\mathrm{i}_k$ such that $\mathrm{i}^2_k = -1$:
\begin{equation}
    \label{eq : definition : defMn}
    \mathbb{M}_n := \{ \eta_1 + \mathrm{i}_n \eta_2 \mid \eta_1, \eta_2 \in \mathbb{M}_{n-1} \}, \qquad \mathbb{M}_0 := \mathbb{R}.
\end{equation}
Each element of the set $\{ \mathrm{i}_1,\ldots,\mathrm{i}_n \}$ is called a \textit{principal unit}, and any multiplication of distinct elements from this set is called a \textit{composite unit} (e.g. $\mathrm{i}_2 \mathrm{i}_4 \mathrm{i}_9$). To get the whole structure of the multicomplex numbers we need to add the following properties:
\begin{enumerate}
    \item $(\mathrm{i}_j \mathrm{i}_k) \mathrm{i}_l = \mathrm{i}_j (\mathrm{i}_k \mathrm{i}_l)$, \hspace{3mm} $\forall j,k,l = 1,\ldots,n$ \hfill (associativity of units)
    \item $\mathrm{i}_j \mathrm{i}_k = \mathrm{i}_k \mathrm{i}_j$, \hspace{3mm} $\forall j,k = 1,\ldots,n$. \hfill (commutativity of units)
\end{enumerate}
It is well known from this construction that we get a commutative, unitary algebra over both $\mathbb{R}$ and $\mathbb{C}$ (considering $\mathrm{i} = \mathrm{i}_1$) with zero divisors when the set $\mathbb{M}_n$ is equipped with term by term addition and multiplication defined in the usual way \cite{GBPrice}. Also, important subspaces trivially isomorphic to the complex numbers are given by $\mathbb{C}(\mathrm{i}_k) = \{ a + b \mathrm{i}_k \mid a,b \in \mathbb{R} \}$. 

The \textit{standard representation} of a multicomplex number $\eta\in \mathbb M_n$ for the first values of $n$ is given by the following expressions

\begin{equation}
\label{standardrep}
\begin{array}{lrcl}
\eta \in \mathbb M_0\simeq \mathbb R: &\eta &=& x_0,  \\*[2ex]
    \eta \in \mathbb M_1\simeq \mathbb C: & \eta &=& x_0 + x_1 \mathrm{i}_1, \\*[2ex]
    \eta \in \mathbb M_2: & \eta &=& x_0 + x_1 \mathrm{i}_1 + x_2 \mathrm{i}_2 + x_{12} \mathrm{i}_1 \mathrm{i}_2, \\*[2ex]
    \eta \in \mathbb M_3: &\eta &=& x_0 + x_1 \mathrm{i}_1 + x_2 \mathrm{i}_2 +  x_3 \mathrm{i}_3 + x_{12} \mathrm{i}_1 \mathrm{i}_2  + x_{13} \mathrm{i}_1 \mathrm{i}_3 \\*[2ex] &&& + \ x_{23} \mathrm{i}_2 \mathrm{i}_3+\ x_{123} \mathrm{i}_1 \mathrm{i}_2 \mathrm{i}_3,\\ & \vdots
\end{array}    
\end{equation}
where the coefficients $x_k$ are real and $\mathbb M_2,\mathbb M_3$ are the bicomplex and the tricomplex algebras, respectively \cite{GBPrice}. One can write the general expression of $\eta\in \mathbb M_n$ in terms of the power set $\mathcal{P} \big(\{ 1,\ldots,n \}\big)$ :
\[ \eta = \sum_{\mathcal A \in \mathcal{P}_n } x_{\mathcal A} \mathrm{i}_{\mathcal A}, \qquad \mathcal{P}_n:=\mathcal{P} \big(\{ 1,\ldots,n \}\big),\qquad x_{\mathcal A}\in \mathbb R, \]
where the empty set in $\mathcal{P}_n$ is associated with the index zero and $\mathrm i_0:=1$, the singleton $\{k\}\in \mathcal{P}_n$ is associated with the index $k$, the set $\{k,l\}\in \mathcal{P}_n$ is associated with the indices $kl$ such that $\mathrm{i}_{kl} := \mathrm{i}_k \mathrm{i}_l$, etc.

\subsection{Conjugation and composition}

A complex-like conjugation $\dagger_k$ is defined by
\begin{equation}
    \dagger_0:=\mathrm{id}\qquad \text{and}\qquad \dagger_k : \mathrm{i}_k \rightarrow - \mathrm{i}_k, \hspace{3mm} k=1,\ldots,n
\end{equation}
where $\dagger_0$ is the identity map. We combine these conjugations with the composition operation denoted $\circ$ and defined as follow :
\begin{equation}
    \eta^{\dagger_j \circ \dagger_k} := (\eta^{\dagger_j})^{\dagger_k}, \qquad \forall \eta \in \mathbb{M}_n \quad\text{ and }\quad 0\leq j,k\leq n.
\end{equation}
The composition is associative and commutative and the set $\ddagger$ of all conjugates with the composition operation is a commutative group $(\ddagger,\ \circ)$ of order $2^n$, where each element is its own inverse and the generators are $\dagger_0,\dagger_1,\ldots,\dagger_n$. This group is isomorphic to $(\mathbb{Z}_2^n,\ +_2)$, see \cite{Garant-Pelletier}. It is easy to show that any composition of conjugates is distributive over addition and multiplication for any multicomplex number. Hence, for $1\leq j_1,j_2,\ldots,j_s\leq n$
the conjugate $\dagger_{j_1}\dagger_{j_2}\cdots \dagger_{j_s}$ applied to a multicomplex number $\eta\in \mathbb M_n$ changes the sign of the principal units $\mathrm i_{j_1},\mathrm i_{j_2},\ldots \mathrm i_{j_s}$ in $\eta$. For instance, taking the  composition of conjugates $\eta^{\dagger_1\dagger_3}$ from an element $\eta \in \mathbb M_3$ in (\ref{standardrep}), we obtain
$$
\eta^{\dagger_1\dagger_3}=x_0 - x_1 \mathrm{i}_1 + x_2 \mathrm{i}_2 -  x_3 \mathrm{i}_3 - x_{12} \mathrm{i}_1 \mathrm{i}_2  + x_{13} \mathrm{i}_1 \mathrm{i}_3 - x_{23} \mathrm{i}_2 \mathrm{i}_3 + x_{123} \mathrm{i}_1 \mathrm{i}_2 \mathrm{i}_3.
$$

\subsection{Standard idempotent representation}

When $n \geq 2$, it is possible to take advantage of the presence of zero divisors and idempotent elements in the multicomplex algebra to get a basis of $\mathbb{M}_n$ over $\mathbb{M}_{n-1}$ such that addition and multiplication between multicomplex numbers are done componentwise. Take
\begin{equation}
\label{gamma}
    \gamma_n := \frac{1}{2}(1 + \mathrm{i}_{n-1} \mathrm{i}_n) \hspace{3mm} \text{ and } \hspace{3mm} \gamma_n' := \frac{1}{2}(1 - \mathrm{i}_{n-1} \mathrm{i}_n).
\end{equation}
Both $\gamma_n$ and $\gamma_n'$ are zero divisors and idempotent elements of $\mathbb{M}_n$, i.e.
\begin{equation}
    \label{eq : idempotent : properties}
         \gamma_n \cdot \gamma_n'  = 0,\qquad \gamma_n^2  = \gamma_n,\qquad (\gamma_n')^2  = \gamma_n'.
\end{equation}
Moreover, we also have the following relations for the multiplicative unit $1$ and the unit $\mathrm{i}_n$ expressed in terms of $\gamma_n$ and $\gamma_n'$:
\begin{equation}
    \label{eq : idempotent : units}
    \begin{aligned}
         1 &= \gamma_n + \gamma_n',\\
         \mathrm{i}_{n} &= -\mathrm{i}_{n-1}(\gamma_n - \gamma_n').
    \end{aligned}
\end{equation}
Now let $\eta \in \mathbb{M}_n$ be a multicomplex number with components $\eta_1, \eta_2 \in \mathbb{M}_{n-1}$. Then by (\ref{eq : idempotent : units}):
\[ \begin{aligned}
    \eta & = \eta_1 + \eta_2 \mathrm{i}_n= \eta_1 (\gamma_n + \gamma_n') - \eta_2 \mathrm{i}_{n-1}(\gamma_n - \gamma_n')\\*[2ex]
    & = (\eta_1 - \eta_2 \mathrm{i}_{n-1}) \gamma_n + (\eta_1 + \eta_2 \mathrm{i}_{n-1}) \gamma_n'.
\end{aligned} \]
Hence $\{ \gamma_n,\ \gamma_n' \}$ forms a basis of $\mathbb{M}_n$ over $\mathbb{M}_{n-1}$. Let $\eta,\zeta \in \mathbb{M}_n$ with corresponding components $\eta_1, \eta_2$ and $\zeta_1, \zeta_2$ in $\mathbb{M}_{n-1}$ relative to that new basis. Then from (\ref{eq : idempotent : properties}), the addition and multiplication operations in $\mathbb{M}_n$ are done componentwise:
\begin{equation}
    \begin{aligned}
    \eta + \zeta &= (\eta_1 \gamma_n + \eta_2 \gamma_n') + (\zeta_1 \gamma_n + \zeta_2 \gamma_n')=(\eta_1+\zeta_1)\gamma_n+(\eta_2+\zeta_2)\gamma_n',\\*[2ex]
    \eta \cdot \zeta &= (\eta_1 \gamma_n + \eta_2 \gamma_n') \cdot (\zeta_1 \gamma_n + \zeta_2 \gamma_n') = \eta_1 \zeta_1 \gamma_n + \eta_2 \zeta_2 \gamma_n'.
    \end{aligned}
\end{equation}
This representation of multicomplex numbers is called the \textit{standard idempotent representation}.

\subsection{Canonical idempotent representation}
$$
$$
Let us introduce the element $\idemProd_n\in \mathbb M_n$ defined as the product of all consecutive idempotent numbers, i.e. 
\begin{equation}
    \idemProd_n := \gamma_2\gamma_3\cdots \gamma_n,\qquad n\geq 2.
\end{equation}
By definition, we have the recursive relation $\idemProd_{n+1} = \idemProd_n \gamma_{n+1}$. We define $\idemProd_n^{\ddagger}$ as the set of all  compositions of conjugates for the element $\idemProd_n$, i.e.
$$
\idemProd_n^{\ddagger}:=\big\{\idemProd_n^{\dagger_{j_1}\dagger_{j_2}\cdots \dagger_{j_k}}\ |\ 0\leq j_1<j_2<\cdots < j_k\leq n,\quad k=1,\ldots,n\big\}.
$$

\begin{proposition}
\label{Prop1}
For $n \geq 2$ the set $\idemProd_n^{\ddagger}$ has $2^{n-1}$ distinct elements. If these elements are represented by $\canBasis_k$ for $k = 1,\ldots,2^{n-1}$ we have
\[ \idemProd_{n+1}^{\ddagger} = \{ \canBasis_k \gamma_{n+1} \mid k=1,\ldots,2^{n-1} \} \cup \{ \canBasis_k \gamma_{n+1}' \mid k=1,\ldots,2^{n-1} \}. \]
\end{proposition}

\begin{proof}
    We proceed by induction over $n$. In the case $n=2$, we have $\idemProd_2 = \gamma_2$, and the values obtained when applying each composition of conjugates are
    \[ \begin{array}{c c c c}
         \idemProd_2^{\dagger_0} = \gamma_2, & \idemProd_2^{\dagger_1} = \gamma_2', & \idemProd_2^{\dagger_2} = \gamma_2', & \idemProd_2^{\dagger_1\dagger_2} = \gamma_2,
    \end{array} \]
    such that $\idemProd_2^{\ddagger}$ contains two distinct elements $\{\canBasis_1=\gamma_2, \canBasis_2=\gamma_2'\}$. Moreover,
    \begin{align*}
    \idemProd_3^{\ddagger}&:=\{\idemProd_3,\idemProd_3^{\dagger_1},\idemProd_3^{\dagger_2},\idemProd_3^{\dagger_3},\idemProd_3^{\dagger_1\dagger_2},\idemProd_3^{\dagger_1\dagger_3},\idemProd_3^{\dagger_2\dagger_3},\idemProd_3^{\dagger_1\dagger_2\dagger_3}\} \\
    &= \{\gamma_2\gamma_3,(\gamma_2\gamma_3)^{\dagger_1},(\gamma_2\gamma_3)^{\dagger_2},(\gamma_2\gamma_3)^{\dagger_3},(\gamma_2\gamma_3)^{\dagger_1\dagger_2},(\gamma_2\gamma_3)^{\dagger_1\dagger_3},(\gamma_2\gamma_3)^{\dagger_2\dagger_3},(\gamma_2\gamma_3)^{\dagger_1\dagger_2\dagger_3}\} \\                            
    &=\{\gamma_2\gamma_3,\gamma_2'\gamma_3,\gamma_2'\gamma_3',\gamma_2\gamma_3',\gamma_2\gamma_3',\gamma_2'\gamma_3',\gamma_2'\gamma_3,\gamma_2\gamma_3\} \\
    &=\{\canBasis_1\gamma_3,\canBasis_2\gamma_3\}\cup \{\canBasis_1\gamma_3',\canBasis_2\gamma_3'\}.
    \end{align*} 
    Assuming the statement to be true in the case $n\geq 2$, we set $\idemProd_n^{\ddagger}=\{ \canBasis_k \}_{k=1}^{2^{n-1}}$. Since $\idemProd_{n+1} = \idemProd_{n} \gamma_{n+1}$ and any composition of conjugates is distributive over the product, then every elements of $\idemProd_{n+1}^{\ddagger}$ can be written as
    \begin{align*}
    \idemProd_{n+1}^{\dagger_{j_1}\cdots \dagger_{j_s}} &=(\idemProd_{n}\gamma_{n+1})^{\dagger_{j_1}\cdots \dagger_{j_s}}\quad \text{for } \quad 0\leq j_1<j_2<\cdots<j_s\leq n+1 \\
    &= (\idemProd_{n})^{\dagger_{j_1}\cdots \dagger_{j_s}}(\gamma_{n+1})^{\dagger_{j_1}\cdots \dagger_{j_s}}\\
    &=\canBasis_k\gamma_{n+1} \quad \mathrm{ or } \quad \canBasis_k\gamma_{n+1}'\quad \text{for } 
    \quad 1\leq k\leq 2^{n-1}.
    \end{align*}
     We verify easily that the pair of multicomplex numbers $\canBasis_k\gamma_{n+1}$ and $\canBasis_l\gamma_{n+1}$ as well as $\canBasis_k\gamma_{n+1}'$ and $\canBasis_l\gamma_{n+1}'$ are distinct for $k,l=1,\ldots, 2^{n-1}$ and $k\neq l$. Moreover,  $\canBasis_k\gamma_{n+1}$ and $\canBasis_l\gamma_{n+1}'$ are also distinct for $k,l=1,\ldots, 2^{n-1}$. Therefore, we can conclude that $\idemProd_{n+1}^{\ddagger}$ has $2 \cdot 2^{n-1}$ distinct elements. 
\end{proof}
\begin{remark} 
\label{changement variables}
 We can rename the elements of $\idemProd_{n+1}^{\ddagger}$ in Proposition~\ref{Prop1} such that this set can be rewritten in the form $\{\tilde{\canBasis}_k\}_{k=1}^{2^n}$. Indeed, for $k=1,\ldots, 2^{n-1}$ we set
 $$
  \tilde{\canBasis}_k:=\canBasis_k\gamma_{n+1}\qquad \text{and}\qquad \tilde{\canBasis}_{2^{n-1}+k}:=\canBasis_k\gamma_{n+1}',
 $$
 where $\{\canBasis_k\}_{k=1}^{2^{n-1}}$ is the set $\idemProd_n^{\ddagger}$.
 
\end{remark}

\begin{proposition}
    \label{prop : canonical : properties}
    For $n \geq 2$ the elements $\{ \canBasis_k \}_{k=1}^{2^{n-1}}$ of $\idemProd_n^{\ddagger}$ have the following properties:
    \[ 
         (i)\ \canBasis_k \canBasis_l = \delta_{kl}\canBasis_l \qquad  (ii)\  \canBasis_k^{\Lambda_n} = \canBasis_k  \qquad (iii)\  \sum_{j=1}^{2^{n-1}} \canBasis_j = 1, 
     \]
    for $1 \leq k,l \leq 2^{n-1}$, where the symbol $\delta_{kl}$ is the usual Kronecker delta and $\Lambda_n:=\dagger_1\dagger_2\cdots \dagger_n$ is the composition of all single conjugates in $\mathbb{M}_n$.
\end{proposition}

\begin{proof}
We proceed by induction over $n$. For $n=2$ we have from the proof of Proposition \ref{Prop1} that $\canBasis_1=\gamma_2$ and $\canBasis_2=\gamma_2'$, hence $\canBasis_k \canBasis_l = \delta_{kl}\canBasis_l$ for $k,l=1,2$. Moreover, $\canBasis_{1}^{\Lambda_2}=\gamma_2^{\dagger_1\dagger_2}=\gamma_2=\canBasis_1$, $\canBasis_{2}^{\Lambda_2}=\canBasis_2$ and $\canBasis_1+\canBasis_2=\gamma_2+\gamma_2'=1$.

Suppose now that all three properties of the proposition are satisfied for $n\geq 2$. Then from Proposition \ref{Prop1} and Remark \ref{changement variables} we have $\idemProd_{n+1}^{\ddagger}=\{\tilde{\canBasis}_k\}_{k=1}^{2^n}$ such that for $1\leq k,l\leq 2^{n-1}$ we obtain 
\begin{align*}
\tilde{\canBasis}_k\cdot \tilde{\canBasis}_l &=\canBasis_k\gamma_{n+1}\cdot \canBasis_l\gamma_{n+1}=\canBasis_k\canBasis_l\gamma_{n+1}=\delta_{kl}\canBasis_l\gamma_{n+1}=\delta_{kl}\tilde{\canBasis}_l,\\
\tilde{\canBasis}_{2^{n-1}+k}\cdot \tilde{\canBasis}_{2^{n-1}+l} &=\canBasis_k\gamma_{n+1}'\cdot \canBasis_l\gamma_{n+1}'=\canBasis_k\canBasis_l\gamma_{n+1}'=\delta_{kl}\canBasis_l\gamma_{n+1}'=\delta_{kl}\tilde{\canBasis}_{2^{n-1}+l}, \\
\tilde{\canBasis}_k\cdot \tilde{\canBasis}_{2^{n-1}+l}&=\canBasis_k\gamma_{n+1}\cdot \canBasis_l\gamma_{n+1}'=0,
\end{align*}
which proves ($i$). Moreover, we have 
$$
\tilde{\canBasis}_k^{\Lambda_{n+1}}=(\canBasis_k\gamma_{n+1})^{\Lambda_{n+1}}=\canBasis_k^{\Lambda_{n}}\gamma_{n+1}^{\Lambda_{n+1}}=\canBasis_k\gamma_{n+1}^{\Lambda_{n+1}}=\canBasis_k(\gamma_{n+1})^{\dagger_{n}\dagger_{n+1}}=\canBasis_k\gamma_{n+1}=\tilde{\canBasis}_k
$$
and $\tilde{\canBasis}_{2^{n-1}+k}^{\Lambda_{n+1}}=\tilde{\canBasis}_{2^{n-1}+k}$ by a similar calculation, which demonstrate ($ii$). Finally, we have
\begin{align*}
\sum_{j=1}^{2^{n}}\tilde{\canBasis}_j&=\sum_{k=1}^{2^{n-1}}\tilde{\canBasis}_k+\sum_{k=1}^{2^{n-1}}\tilde{\canBasis}_{2^{n-1}+k}=
\sum_{k=1}^{2^{n-1}} \canBasis_k \gamma_{n+1} + \sum_{k=1}^{2^{n-1}} \canBasis_k \gamma_{n+1}' \\
& =  (\gamma_{n+1} + \gamma_{n+1}') \sum_{k=1}^{2^{n-1}} \canBasis_k = 1
\end{align*}
which shows ($iii$).
\end{proof}
These properties assure us that the distinct elements $\{ \canBasis_k \}_{k=1}^{2^{n-1}}$ of $\idemProd_n^{\dagger}$ are linearly independent over $\mathbb{C}$. Let $\alpha_1,\ldots,\alpha_{2^{n-1}} \in \mathbb{C}$, then 
    $\sum_{k=1}^{2^{n-1}} \alpha_k \canBasis_k = 0$ iff $\alpha_j \canBasis_j=0$ for all $j$, i.e. $\alpha_j =0$.

\begin{proposition}
    The set $\idemProd_n^{\ddagger}$ is a basis of $\mathbb{M}_n$ over $\mathbb{C}$.
\end{proposition}

\begin{proof}
    We proceed by induction over $n$. In the case $n=2$, we know the statement is true since $\idemProd_2^{\ddagger} = \{ \gamma_2,\gamma_2' \}$ is the basis for the standard idempotent representation of $\mathbb{M}_2$ over $\mathbb{M}_1 \simeq \mathbb{C}$.
    Suppose the statement is true for $n\geq 2$, i.e. any number 
    $\zeta \in \mathbb{M}_n$ can be written in the form $\zeta = \sum_{k=1}^{2^{n-1}} z_k {\canBasis_k},\ z_k \in \mathbb{C}.$ Here $\{ \canBasis_k \}_{k=1}^{2^{n-1}}$ are the distinct elements of $\idemProd_n^{\ddagger}$. Let $\eta \in \mathbb{M}_{n+1}$. From the standard idempotent representation we have
    \begin{equation}
        \label{canRep : etaStandardIdempRepEq}
        \eta = \zeta_1 \gamma_{n+1} + \zeta_2 \gamma_{n+1}', \qquad \zeta_1, \zeta_2 \in \mathbb{M}_n.
    \end{equation}
    The induction hypothesis implies that  
    $\zeta_1 = \sum_{k=1}^{2^{n-1}} \alpha_k \canBasis_k$ and $\zeta_2 = \sum_{k=1}^{2^{n-1}} \beta_k \canBasis_k$ for $\alpha_k,\beta_k\in \mathbb C$.
    By substituing in (\ref{canRep : etaStandardIdempRepEq}),
    \[ \eta = \sum_{k=1}^{2^{n-1}} \alpha_k \canBasis_k \gamma_{n+1} + \sum_{k=1}^{2^{n-1}} \beta_k \canBasis_k \gamma_{n+1}'=\sum_{k=1}^{2^{n-1}} \alpha_k \tilde{\canBasis}_k  + \sum_{k=1}^{2^{n-1}} \beta_k \tilde{\canBasis}_{2^{n-1}+k}, 
     \]
     where the change of variables of Remark~\ref{changement variables} is used.
     Hence, the set $\idemProd_{n+1}^{\ddagger}=\{\tilde{\canBasis}_k\}_{k=1}^{2^{n}}$ is a basis of $\mathbb{M}_{n+1}$ over $\mathbb{C}$.
\end{proof}

\begin{theorem}[Canonical idempotent representation] For $n\geq 2$, any multicomplex element $\eta\in \mathbb M_n$ can be represented by 
$$
\eta=\sum_{k=1}^{2^{n-1}} z_k \canBasis_{k},\qquad z_k\in \mathbb C(\mathrm i_i)
$$
and the $2^{n-1}$ idempotent elements $\canBasis_{k}$ satisfy Proposition \ref{prop : canonical : properties}.
\end{theorem}

\begin{remark}
To simplify the notation for the rest of the paper we will consider element $\mathrm i_1$ as the usual imaginary complex number $\mathrm i$. Therefore, we consider $\mathbb C(\mathrm i_i)\simeq \mathbb C$ and $\eta=\sum_{k=1}^{2^{n-1}} z_k \canBasis_{k} $ where $z_k\in \mathbb C$ in the last theorem. Moreover for any $z\in \mathbb C(\mathrm i_i)\simeq \mathbb C$ we have $z^{\dagger_1}=\overline z$, where the bar represents the usual complex conjugation.
\end{remark}

\subsection{Projections}

For any multicomplex number $\eta=\sum_{k=1}^{2^{n-1}}z_k\canBasis_k\in \mathbb M_n$ written in the canonical idempotent representation,  we introduce the $j$th multicomplex projection as the function $P_j : \mathbb{M}_n \rightarrow \mathbb{C}$ such that $P_j(\eta) = z_j$. 
In what follows, the caret notation $\hat{\jmath}$ for indices will be used in relation with the $j$th projection in the following manner :
\begin{equation}
\label{canonicalrep}
    \eta_{\hat{\jmath}} := P_j(\eta) \hspace{3mm} \text{ and } \hspace{3mm} \eta = \sum_{k=1}^{2^{n-1}} \eta_{\hat{k}} \canBasis_k.
\end{equation}
In particular, for any $z\in \mathbb C$ we have from Proposition~\ref{prop : canonical : properties} (iii) that
$$z=\sum_{k=1}^{2^{n-1}} z\canBasis_k \quad \Rightarrow \quad P_j(z)=z\quad \text{for}\quad 1\leq j\leq 2^{n-1},
$$
i.e. the $j$th projection is the identity map when applied on complex elements.
The projection operator is a ring-homomorphism of the multicomplex numbers, i.e. for all $\eta,\zeta\in \mathbb M_n$
$$
P_j(\eta+\zeta)=P_j(\eta)+P_j(\zeta)\quad \text{and}\quad P_j(\eta\cdot\zeta)=P_j(\eta)\cdot P_j(\zeta).
$$
A multicomplex number $\eta$ is a zero divisor if and only if at least one of its projections vanishes. Indeed, for two non zero elements $\eta,\zeta\in \mathbb M_n$ such that $\eta\zeta=0$, $\eta=\sum_{k=1}^{2^{n-1}} \eta_{\hat{k}} \canBasis_k$ and $\zeta=\sum_{k=1}^{2^{n-1}} \zeta_{\hat{k}} \canBasis_k$ we have
$$
\eta\zeta=\sum_{k=1}^{2^{n-1}} \eta_{\hat{k}}\zeta_{\hat{k}} \canBasis_k=0 \quad \Rightarrow \quad \eta_{\hat{k}}=0 \text{ or }\zeta_{\hat{k}}=0 \quad \text{for}\quad 1\leq k\leq 2^{n-1}.
$$
Conversely, we have
\begin{equation}
    \eta_{\hat{\jmath}} = 0 \quad \Leftrightarrow \quad \canBasis_j \eta = \canBasis_j \eta_{\hat{\jmath}} = 0
\end{equation}
which holds even if $\eta \neq 0$. The set of zero divisors including zero, denoted by $\mathbb M_n^{-1}$, is then represented by
$$
\mathbb M_n^{-1}=\big\{\eta=\sum_{k=1}^{2^{n-1}}\eta_{\hat k}\canBasis_k\ |\ \eta_{\hat{\jmath}}=0 \text{ for at least one }j\big\}.
$$ 
For any $\eta=\sum_{k=1}^{2^{n-1}}\eta_{\hat k}\canBasis_k\in \mathbb M_n\backslash \mathbb M_n^{-1}$ the formula for the inverse is $\eta^{-1}=\sum_{k=1}^{2^{n-1}}\eta^{-1}_{\hat k}\canBasis_k$.


\subsection{Multiperplex subalgebra}

An important subalgebra called the multiperplex numbers and denoted $\mathbb{D}_n$ is the one made up by multicomplex numbers invariant under the $\Lambda_n$ composition of conjugates :
\[ \mathbb{D}_n = \{ \zeta \in \mathbb{M}_n \mid \zeta^{\Lambda} = \zeta \}. \]
\begin{remark}
For the rest of this article we will omit index $n$ in $\Lambda$ except where necessary for clarity.
\end{remark}

Since all $\canBasis_k$ from the canonical representation are invariant under $\Lambda$ (see Proposition \ref{prop : canonical : properties}), $\idemProd_n^{\ddagger}$ is a basis of $\mathbb{D}_n$. Expanding the equality $\zeta^{\Lambda} = \zeta$ we can see that a multicomplex number is in the subalgebra $\mathbb{D}_n$ if and only if all its components in the canonical representation are real :
\begin{equation}
        \zeta^{\Lambda} = \zeta  \quad\Leftrightarrow \quad \sum_{k=1}^{2^{n-1}} \overline{\zeta_{\hat{k}}} \canBasis_k = \sum_{k=1}^{2^{n-1}} \zeta_{\hat{k}} \canBasis_k
         \quad \Leftrightarrow  \quad \overline{\zeta_{\hat{k}}} = \zeta_{\hat{k}},
\end{equation}
i.e. $\zeta_{\hat k}\in \mathbb R$ for $k=1,\ldots,2^{n-1}$.

As a result, we obtain a broader notion of real and imaginary parts by separating each multicomplex element into two parts contained in $\mathbb{D}_n$. Let $\eta \in \mathbb{M}_n$, then each component $\eta_{\hat{k}}$ of $\eta$ in the canonical idempotent representation can be written as $\eta_{\hat{k}} = x_k + \mathrm{i} y_k$ for $x_k,y_k \in \mathbb{R}$:
\begin{equation}
\label{perplexrepofmulticomplex}
    \eta = \sum_{k=1}^{2^{n-1}} (x_k + \mathrm{i} y_k) \canBasis_k = \sum_{k=1}^{2^{n-1}} x_k \canBasis_k + \mathrm{i} \sum_{k=1}^{2^{n-1}} y_k \canBasis_k = d_1 + \mathrm{i} d_2,
\end{equation}
where $d_1, d_2 \in \mathbb{D}_n$. The set of multiperplex numbers is a vector space over $\mathbb{R}$ and for any $\eta,\zeta \in \mathbb{D}_n$ we can define the following vector partial order $\leq$ :
$$ \eta \leq \zeta \quad \Leftrightarrow \quad \eta_{\hat{\jmath}} \leq \zeta_{\hat{\jmath}} \qquad \text{for}\qquad  j = 1, \ldots, 2^{n-1}.
$$
This is also known as the product order (or componentwise order) from \cite{sudhir2010Multivariable}. We denote $\mathbb{D}_n^+$ the set of all positive multiperplex numbers i.e.
$$ 
\mathbb{D}_n^+ = \{ \eta \in \mathbb{D}_n\ |\ \eta \geq 0 \}. 
$$
Let $\eta \in \mathbb{M}_n$. A natural and unique way to define a multiperplex-valued norm on $\mathbb{M}_n$ from the $n$ composition of conjugates $\Lambda=\dagger_1\cdots \dagger_n$ is
$$ \Vert \eta \Vert := \sqrt{\eta^{\Lambda} \eta}. $$
As the multiplication acts componentwise in the canonical basis, the square root is distributed over the components such that
$$ \Vert \eta \Vert := \sqrt{\eta^{\Lambda} \eta} =  \sum_{k=1}^{2^{n-1}} \sqrt{|\eta_{\hat{k}}|^2} \canBasis_k = \sum_{k=1}^{2^{n-1}} |\eta_{\hat{k}}| \canBasis_k \in \mathbb{D}_n^+ $$
and the result is a positive multiperplex number since $|\eta_{\hat{\jmath}}| \geq 0$ for all $j = 1,\ldots,2^{n-1}$. 
For all $\eta, \zeta \in \mathbb{M}_n$ we have
$$ \Vert \eta + \zeta \Vert = \sum_{k=1}^{2^{n-1}} |\eta_{\hat{k}} + \zeta_{\hat{k}}| \canBasis_k \leq \sum_{k=1}^{2^{n-1}} |\eta_{\hat{k}}| \canBasis_k + \sum_{k=1}^{2^{n-1}} |\zeta_{\hat{k}}| \canBasis_k = \Vert \eta \Vert + \Vert \zeta \Vert $$
since $|\eta_{\hat{\jmath}} + \zeta_{\hat{\jmath}}| \leq |\eta_{\hat{\jmath}}| + |\zeta_{\hat{\jmath}}|$ for all $j$. The other norm properties (absolute homogeneity and positivity) can be verified in the same way as for the complex norm. The multiperplex-valued norm defined on $\mathbb{M}_n$ is a generalization of the hyperbolic-valued norm defined on the bicomplex numbers ($\mathbb{BC} \simeq \mathbb{M}_2$) from \cite{Luna-Elizarrarás_Perez-Regalado_Shapiro_2014}.

The multicomplex algebra $\mathbb{M}_n$ equipped with this norm and the conjugate $\Lambda$ as the involution  satisfies all the properties of a $C^*$-algebra (see \cite{JBConway}) except one: the norm is multiperplex-valued rather than real. This fact comes once again from the properties of complex numbers combined to the properties of $\Lambda$ as a composition of conjugates.

\section{Multicomplex ideals}
\subsection{Multiperplex ideals}

A multicomplex ring ideal is a subring $I$ of $\mathbb{M}_n$ such that the set $\eta I := \{ \eta\,\zeta \mid \zeta \in I \}$ is contained in $I$ for all $\eta \in \mathbb{M}_n$ (the multiperplex ring ideal is defined in the same way on $\mathbb{D}_n$). Starting with the description of multiperplex ideals, we will see later that we can get the multicomplex ideals from the complexification operation on vector spaces, thus covering the study of both at the same time.

From properties (i) and (iii) of Proposition \ref{prop : canonical : properties}, the basis elements $\idemProd_n^{\ddagger}=\{ \canBasis_k \}_{k=1}^{2^{n-1}} $ of the canonical idempotent representation are an \textit{orthogonal decomposition of the identity} (defined in \cite{Hazewinkel} page 30). This decomposition allow us to write the multiperplex ring as the following finite direct sum:
\begin{equation}
    \mathbb{D}_n = \bigoplus_{k=1}^{2^{n-1}} \mathbb{D}_n \canBasis_k.
\end{equation}
Each term $\mathbb{D}_n \canBasis_j$ of this sum is the principal ideal generated by the corresponding single element $\canBasis_j$. Since this decomposition is also a basis of the space, for all $\eta \in \mathbb{D}_n$ we have $\eta\,\canBasis_j = \big( \sum_{k=1}^{2^{n-1}} \eta_{\hat{k}} \canBasis_k \big) \canBasis_j = \eta_{\hat{\jmath}}\,\canBasis_j$
and $\mathbb{D}_n \canBasis_j = \mathbb{R} \canBasis_j$, meaning that all principal ideals generated by an element of the canonical idempotent basis are of the form $\mathbb{R} \canBasis_j$ and the space $\mathbb{D}_n$ is a direct sum of these ideals. We will show that not only $\mathbb{D}_n$, but all multiperplex ideals can be characterized in the same way.
\begin{proposition}
    Let $I\neq \{0\}$ be a proper multiperplex ideal of $\mathbb D_n$. Then all non-zero elements of $I$ are zero divisors.
\end{proposition}
\begin{proof}
    Let $\eta \in I$ be a non-zero element. Suppose that $\eta$ is not a zero divisor, then $\eta^{-1}\,\eta = 1 \in I$. Hence the presence of the identity element in $I$ implies $I = \mathbb{D}_n$, a contradiction.
\end{proof}

\begin{proposition}
    The principal ideal $\mathbb{R}\canBasis_j$ is minimal for $j = 1,\ldots,2^{n-1}$.
\end{proposition}
\begin{proof}
    Let $I \subset \mathbb{R}\canBasis_j$ be a subideal. Either $I = \{ 0 \}$ or there exists a real number $x \neq 0$ such that $x \canBasis_j \in I$. Then for all $y \in \mathbb{R}$, $(y/x)\cdot x \canBasis_j = y \canBasis_j \in I$ and $I = \mathbb{R}\canBasis_j$.
\end{proof}
\begin{lemma}
    \label{lemma : multiperplexIdeals : containsMinimal}
    A nontrivial multiperplex ideal contains at least one minimal ideal $\mathbb{R}\canBasis_j$.
\end{lemma}
\begin{proof}
    Let $I\neq \{0\}$ be a multiperplex ideal. Then there exists a non-zero element $\eta \in I$ with at least one non-vanishing projection $\eta_{\hat{\jmath}}$. Then for all $x \in \mathbb{R}$
    \[ \eta \, \frac{x \canBasis_j}{\eta_{\hat{ \jmath}}} = \frac{x \eta_{\hat{\jmath}} \canBasis_j}{\eta_{\hat{\jmath}}} = x \canBasis_j \in I \quad \Rightarrow \quad \mathbb{R} \canBasis_j \subseteq I. \]
\end{proof}
\begin{lemma}
    \label{lemma : multiperplexIdeals : directIsIdeal}
    For any subset of indices $J \subseteq \{ 1,\ldots,2^{n-1} \}$, the direct sum $\displaystyle \bigoplus_{j \in J} \mathbb{R}\canBasis_j$ is a multiperplex principal ideal generated by the element $\displaystyle \sum_{j \in J} \canBasis_j$.
\end{lemma}
\begin{proof}
    We have directly that
    \[ \begin{aligned}
        \mathbb{D}_n \Big( \sum_{j \in J} \canBasis_j \Big) & = \Big\{ \eta \sum_{j \in J} \canBasis_j \mid \eta \in \mathbb{D}_n \Big\}= \Big\{ \sum_{j \in J} \eta_{\hat{\jmath}} \canBasis_j \mid \eta_{\hat{\jmath}} \in \mathbb{R} \Big\}= \bigoplus_{j \in J} \mathbb{R} \canBasis_j.
    \end{aligned} \]
\end{proof}
\begin{theorem}
    \label{thm : multiperplexIdeals : directSum}
    A multiperplex ideal $I\neq \{0\}$ is a principal ideal of the form $\displaystyle\bigoplus_{j \in J} \mathbb{R}\canBasis_j$, where $J \subseteq \{ 1,\ldots,2^{n-1} \}$.
\end{theorem}
\begin{proof}
    Let $I$ be a nontrivial multiperplex ideal and $\bigoplus_{j \in J} \mathbb{R} \canBasis_j\subseteq I$ the largest principal ideal  contained in $I$ (from Lemma \ref{lemma : multiperplexIdeals : containsMinimal}, $I$ contains at least one $\mathbb{R} \canBasis_k$). If $I \backslash \bigoplus_{j \in J} \mathbb{R} \canBasis_j = \varnothing$ then $I = \bigoplus_{j \in J} \mathbb{R} \canBasis_j$ which is the desired result. If $I \backslash \bigoplus_{j \in J} \mathbb{R} \canBasis_j \neq \varnothing$, then (since $0\in \bigoplus_{j \in J} \mathbb{R} \canBasis_j$) there exists a non-zero element $\eta \in I \backslash \bigoplus_{j \in J} \mathbb{R} \canBasis_j$ with a non vanishing projection $\eta_{\hat{k}}$. The existence of this element in the ideal implies that $\mathbb{R} \canBasis_{k} \subseteq I$ and
    \[ \Big( \bigoplus_{j \in J} \mathbb{R} \canBasis_j \Big) \oplus \mathbb{R}\canBasis_{k} \subseteq I. \]
    Contradicting the fact that $\bigoplus_{j \in J} \mathbb{R} \canBasis_j$ is the largest direct sum contained in $I$. Thus $I = \bigoplus_{j \in J} \mathbb{R} \canBasis_j$ and is a principal ideal from Lemma \ref{lemma : multiperplexIdeals : directIsIdeal}.
\end{proof}
Having characterized the multiperplex ideals in terms of the minimal ideals, we turn our attention to the orthogonal hyperplanes $H_j$ defined as
\begin{equation}
    H_j := \{ \eta \in \mathbb{D}_n \mid \eta \, \canBasis_j = 0 \}.
\end{equation}
This set corresponds to all multiperplex numbers with a vanishing $j$th projection. We show easily  that $H_j$ is an ideal of the multiperplex numbers since $\zeta\, H_j\subseteq H_j$ for all $\zeta\in \mathbb D_n$.
\begin{proposition}
    For $j=1,\ldots,2^{n-1}$ the orthogonal hyperplane $H_j$ is a maximal multiperplex ideal.
\end{proposition}
\begin{proof}
    Suppose $H_j$ is not maximal, then there is a proper ideal $I$ such that $H_j \subset I$. Let $\eta \in I \backslash H_j$, then $\eta_{\hat{\jmath}} \neq 0$ (otherwise $\eta \in H_j$). This implies that $\mathbb{R} \canBasis_j \subset I$ and
    \[ H_j \oplus \mathbb{R} \canBasis_j = \mathbb{D}_n \subseteq I. \]
    Thus we conclude that all orthogonal hyperplanes are maximal ideals.
\end{proof}
\begin{proposition}
    For any proper subset of indices $\varnothing \neq J \subset \{ 1,\ldots,2^{n-1} \}$, we have
    \[ \bigoplus_{j \in J} \mathbb{R} \canBasis_j = \bigcap_{k \in J^{\complement}} H_{k}. \]
\end{proposition}
\begin{proof}
    We have directly
    \[ 
    \begin{array}{rcl}
    \eta \in \displaystyle\bigoplus_{j \in J} \mathbb{R} \canBasis_j & \Leftrightarrow &  \eta_{\hat{k}} = 0,\ \forall k \in J^{\complement}\quad \Leftrightarrow \quad \eta \in H_{k},\ \forall k \in J^{\complement} \\*[2ex]
    &\Leftrightarrow & \eta \in \displaystyle \bigcap_{k \in J^{\complement}} H_{k}. 
    \end{array}
    \]
\end{proof}
From these two last propositions and Theorem \ref{thm : multiperplexIdeals : directSum}, we see that all multiperplex ideals are written in terms of a direct sum of minimal ideals $\mathbb{R} \canBasis_j$ as well as  intersections of maximal ideals $H_j$. Moreover, there are no other ideals distinct from the $\mathbb{R} \canBasis_j$ and $H_j$ that are also minimal or maximal and the characterization of multiperplex is thus completed.

\subsection{Complexification and realization}

Let $\mathcal{I}(\mathbb{D}_n)$ and $\mathcal{I}(\mathbb{M}_n)$ be the respective sets of multiperplex and multicomplex ideals. The \textbf{complexification} is always applied on a vector space (or algebra) $V$ over $\mathbb{R}$ as the extension of the scalar multiplication over $\mathbb{C}$.
\begin{equation}
    (V, +, \cdot, \mathbb{R}) \xrightarrow[]{\text{Complexification}} (V \oplus \mathrm{i} V, +, \cdot, \mathbb{C})
\end{equation}
Both bases and dimension of $V$ are preserved by the complexification.
\begin{proposition}
    The complexification of a multiperplex ideal is a multicomplex ideal.
\end{proposition}
\begin{proof}
    Let $I_D \in \mathcal{I}(\mathbb{D}_n)$. Considering the complexification $I_D \oplus \mathrm{i} I_D$ of $I_D$, we suppose that $d_1 + \mathrm{i} d_2 \in I_D \oplus \mathrm{i} I_D$. Then, from (\ref{perplexrepofmulticomplex}) any $\eta \in \mathbb{M}_n$ can be written as $\eta = d_1' + \mathrm{i} d_2'$ where $d_1',d_2' \in \mathbb{D}_n$. We find 
    \[ \eta \,(d_1 + \mathrm{i} d_2) = (d_1' + \mathrm{i} d_2') \, (d_1 + \mathrm{i} d_2) = (d_1' d_1 - d_2' d_2) + \mathrm{i} (d_1' d_2 + d_2' d_1). \]
    Since $I_D$ is an ideal, all products of the right-hand side of the equation are in $I_D$ and $\eta \,(d_1 + \mathrm{i} d_2) \in I_D \oplus \mathrm{i} I_D$.
\end{proof} 
From this proposition, the complexification can be seen as a well-defined function from $\mathcal{I}(\mathbb{D}_n)$ to $\mathcal{I}(\mathbb{M}_n)$. 

Let us now define the set operator $\mathcal R$ by
\begin{equation}
    \mathcal{R}(A) := \{ \eta \in A \mid \eta^{\Lambda} = \eta \},\qquad \forall A\subseteq \mathbb M_n.
\end{equation}
We note that for any multicomplex ideal $I_M$ (or multicomplex subring) then $\mathcal{R}(I_M)$ is nonempy since $0 \in I_M$ and $0^{\Lambda} = 0$.
\begin{proposition}
    Let $I_M$ be a multicomplex ideal. Then $\mathcal R(I_M)$ is a multiperplex ideal such that $\mathcal R$ is a mapping from the set of multicomplex ideals to the set of multiperplex ideals, i.e. $\mathcal R:\mathcal I(\mathbb M_n)\rightarrow \mathcal I(\mathbb D_n)$.
\end{proposition}
\begin{proof}
    Let $I_M \in \mathcal{I}(\mathbb{M}_n)$. Then for all $\zeta \in \mathbb{D}_n$ and $\eta\in \mathcal R(I_M)$, we have
    $$
(\zeta\,\eta)^\Lambda=\zeta^\Lambda\,\eta^\Lambda=\zeta\,\eta
    $$
    such that $\zeta \mathcal R(I_M)\subseteq \mathcal R(I_M)$.
\end{proof}
\begin{proposition}
    \label{prop : complexification : ConjugateInIdeal}
    Let $I_M$ be a multicomplex ideal. If $\eta \in I_M$, then $\eta^{\Lambda} \in I_M$.
\end{proposition}
\begin{proof}
    Let $\eta \in I_M$ written in the canonical representation :
    \[ \eta = \displaystyle\sum_{k=1}^{2^{n-1}} \eta_{\hat{k}} \canBasis_k \]
    and take $\zeta \in \mathbb{M}_n$ such that
    \[ 
    P_k(\zeta) = \zeta_{\hat{k}} = \begin{cases}
        0 & \text{ if } \eta_{\hat{k}} = 0,\\*[2ex]
        \overline{\eta_{\hat{k}}}^2/|\eta_{\hat{k}}|^2 & \text{ if } \eta_{\hat{k}} \neq 0,
    \end{cases} 
    \]
    for $k=1,\ldots, 2^{n-1}$.
    Then for all $k$,
    \[ P_k(\zeta \,\eta) = \zeta_{\hat{k}} \,\eta_{\hat{k}} = \begin{cases}
        0 & \text{ if } \eta_{\hat{k}} = 0\\*[2ex]
        \overline{\eta_{\hat{k}}} & \text{ if } \eta_{\hat{k}} \neq 0\\
    \end{cases} \quad\Rightarrow \quad P_k(\zeta \, \eta) = \overline{\eta_{\hat{k}}}. \]
    Thus $\zeta \, \eta = \eta^{\Lambda}$ and since $I_M$ is an ideal, $\eta^{\Lambda} \in I_M$.
\end{proof}
\begin{lemma}
    \label{lemma : complexification : RealizationAsInverse}
     The mapping $\mathcal R:\mathcal{I}(\mathbb{M}_n)\rightarrow \mathcal{I}(\mathbb{D}_n)$ is the inverse of the complexification from $\mathcal{I}(\mathbb{D}_n)$ to $\mathcal{I}(\mathbb{M}_n)$, i.e. for all $I_D \in \mathcal{I}(\mathbb{D}_n)$ and $I_M \in \mathcal{I}(\mathbb{M}_n)$
    \[ \mathcal{R}(I_D \oplus \mathrm{i} I_D) = I_D \qquad \text{ and } \qquad \mathcal{R}(I_M) \oplus \mathrm{i} \mathcal{R}(I_M) = I_M. \]
\end{lemma}
\begin{proof}
    Let $\eta = d_1 + \mathrm{i} d_2 \in I_D \oplus \mathrm{i} I_D$, then
    \[ 
    \begin{array}{rcl}
        \eta \in \mathcal{R}(I_D \oplus \mathrm{i} I_D)  &\Leftrightarrow & \eta^{\Lambda} = \eta \ 
         \Leftrightarrow \ (d_1+ \mathrm{i}d_2)^{\Lambda} = d_1 + \mathrm{i} d_2 \\*[2ex]
         &\Leftrightarrow  & d_2 = 0 \text{ and } \eta \in I_D.
         \end{array}
     \]
    Now, let $\zeta = d_1' + \mathrm{i} d_2' \in \mathcal{R}(I_M) \oplus \mathrm{i} \mathcal{R}(I_M)$. Since $\mathcal{R}(I_M)\subseteq I_M$ from the definition of $\mathcal R$, $\zeta$ is in $I_M$. Conversely, if $\zeta = d_1 + \mathrm{i}d_2 \in I_M$ then from Proposition \ref{prop : complexification : ConjugateInIdeal}, $\zeta^{\Lambda} = d_1 - \mathrm{i} d_2 \in I_M$ and
    \[ \frac{\zeta + \zeta^{\Lambda}}{2} = d_1 \in I_M, \qquad \frac{\zeta - \zeta^{\Lambda}}{2 \mathrm{i}} = d_2 \in I_M. \]
    Since $d_1,d_2$ are multiperplex elements in $I_M$, they are invariant under $\Lambda$, thus $d_1,d_2\in \mathcal R(I_M)$ and $\zeta \in \mathcal{R}(I_M) \oplus \mathrm{i} \mathcal{R}(I_M)$. 
\end{proof}

From ring theory, the intersection of two ideals is an ideal and this operation defines an algebraic structure on $\mathcal{I}(\mathbb{D}_n)$ and $\mathcal{I}(\mathbb{M}_n)$. More specifically, $(\mathcal{I}(\mathbb{D}_n),\ \cap)$ and $(\mathcal{I}(\mathbb{M}_n),\ \cap)$ are both monoïds with respective identity elements $\mathbb{D}_n$ and $\mathbb{M}_n$.

\begin{theorem}
    \label{thm : complexification : oneToOneAndInclusion}
    The complexification is a one-to-one correspondence from $\mathcal{I}(\mathbb{D}_n)$ to $\mathcal{I}(\mathbb{M}_n)$ preserving the intersection, i.e. for any $I_1, I_2 \in \mathcal{I}(\mathbb{D}_n)$,
    \[ (I_1 \cap I_2) \oplus \mathrm{i} (I_1 \cap I_2) = (I_1 \oplus \mathrm{i} I_1) \cap (I_2 \oplus \mathrm{i} I_2). \]
\end{theorem}
\begin{proof}
    The existence of the inverse from Lemma \ref{lemma : complexification : RealizationAsInverse} is sufficient to conclude that the complexification  is bijective between $\mathcal{I}(\mathbb{D}_n)$ and $\mathcal{I}(\mathbb{M}_n)$. If $I_1$ and $I_2$ are two multiperplex ideals then $\eta = d_1 + \mathrm{i} d_2 \in (I_1 \cap I_2) \oplus \mathrm{i} (I_1 \cap I_2)$ if and only if
    \[ 
    \begin{array}{rcl}
        d_1, d_2 \in I_1 \cap I_2
        &\Leftrightarrow &  d_1, d_2 \in I_1 \text{ and } d_1, d_2 \in I_2\\*[2ex]
        &\Leftrightarrow &  d_1 + \mathrm{i} d_2 \in I_1 \oplus \mathrm{i} I_1 \text{ and } d_1 + \mathrm{i} d_2 \in I_2 \oplus \mathrm{i} I_2 \\*[2ex]
        &\Leftrightarrow & \eta \in (I_1 \oplus \mathrm{i} I_1) \cap (I_2 \oplus \mathrm{i} I_2).
        \end{array}
     \]
\end{proof}
Note that the inclusion between ideals is also preserved since $I_1 \subseteq I_2 \Leftrightarrow I_1 \cap I_2 = I_1$. From Theorem \ref{thm : complexification : oneToOneAndInclusion}, for all $I_1, I_2 \in  \mathcal{I}(\mathbb{D}_n)$,
\begin{equation}
    I_1 \subseteq I_2 \quad \Leftrightarrow \quad I_1 \oplus \mathrm{i} I_1 \subseteq I_2 \oplus \mathrm{i} I_2.
\end{equation}

\subsection{Translation to the multicomplex ring}

From Theorem \ref{thm : complexification : oneToOneAndInclusion}, we get the multicomplex minimal and maximal ideals from the complexification of $\mathbb{R} \canBasis_j$ and $H_j$.
\begin{equation}
    \begin{aligned}
        \mathbb{R} \canBasis_j & \rightarrow \mathbb{R} \canBasis_j \oplus \mathrm{i} \mathbb{R} \canBasis_j = \mathbb{C} \canBasis_j,\\ 
        H_j & \rightarrow H_j \oplus \mathrm{i} H_j.
    \end{aligned}
\end{equation}
The complexification of $H_j$ is itself a multicomplex hyperplane denoted $E_j$ where
\begin{equation}
    E_j = H_j \oplus \mathrm{i} H_j = \{ \eta \in \mathbb{M}_n \mid \eta \, \canBasis_j=0 \}.
\end{equation}
Indeed, for all $\eta = d_1 + \mathrm{i} d_2 \in \mathbb{M}_n$, $d_1\canBasis_j = 0$ and $d_2 \canBasis_j = 0$ if and only if $\eta \,\canBasis_j = 0$. In the same way that $\mathbb{R} \canBasis_j$ and $H_j$ can be seen as the building blocks of multiperplex ideals, a complete characterization of multicomplex ideals is done in terms of $\mathbb{C} \canBasis_j$ and $E_j$.
\begin{corollary}
    All multicomplex ideals are principal ideals of the form
    \[ \bigoplus_{j \in J} \mathbb{C} \canBasis_j = \bigcap_{k \in J^{\complement}} E_k \]
    generated by the elements $\displaystyle \sum_{j \in J} \canBasis_j$, where $\varnothing \neq J \subseteq \{ 1,\ldots,2^{n-1}\}$.
\end{corollary}
Taking any multicomplex ideal $I_J$ with a given subset of indices $J$, the corresponding quotient ring $\mathbb{M}_n/I_J$ is written as
\[ \mathbb{M}_n/I_J = \{ \eta + I_J \mid \eta \in \mathbb{M}_n \} = \{ \zeta + I_J \mid \zeta \in I_{J^{\complement}} \} \simeq I_{J^{\complement}}. \]
The last part comes from the canonical representation which let us separate $\eta$ into two distinct sums : one contained in $I_J$ and the other one in $I_{J^{\complement}}$.

\section{Free $\mathbb{M}_n$-module}

\subsection{Bases and subspaces}
From \cite{rotman2002advanced}, any free module over a commutative ring has a well-defined rank. As $\mathbb{M}_n$ is a commutative algebra, any two basis sets of a free $\mathbb{M}_n$-module have the same (fixed) cardinality $m$. 


Let $W$ be a free $\mathbb{M}_n$-module and $\{ \ket{w_l} \}_{l=1}^m$ a basis made up of $m < \infty$ elements. This describes the finite-dimensional case, and an element $\ket{\psi}$ of $W$ is written as a linear combination of the basis elements :
\begin{equation}
    \label{eq:generalKetOfModule}
    \ket{\psi} = \sum_{l=1}^m \eta_l \ket{w_l}, \qquad \eta_l \in \mathbb{M}_n.
\end{equation}
An important subset $V \subset W$ is the set of all elements with coefficients restricted to the field of complex numbers 
\begin{equation}
\label{defV}
    V := \Big\{ \sum_{l=1}^m z_l \ket{w_l} \mid z_l \in \mathbb{C} \Big\}.
\end{equation}
This space $V$ is a $m$-dimensional vector space over $\mathbb{C}$, see \cite{BCQM1}. We can use the canonical idempotent representation of multicomplex numbers to further develop the expression of any $\ket{\psi}$ in $W$. For each $\eta_l\in \mathbb M_n$ we have :
\begin{equation}
    \eta_l = \sum_{k=1}^{2^{n-1}} z_{l,\hat{k}} \canBasis_k, \qquad z_{l,\hat{k}} \in \mathbb{C}.
\end{equation}
Substituting in (\ref{eq:generalKetOfModule}), we get
\begin{equation}
    \label{eq:ketDoubleSum}
    \ket{\psi} = \sum_{l=1}^m \Big( \sum_{k=1}^{2^{n-1}} z_{l,\hat{k}} \canBasis_k \Big) \ket{w_l} = \sum_{k=1}^{2^{n-1}} \canBasis_k \sum_{l=1}^m z_{l,\hat{k}} \ket{w_l}
\end{equation}
and we set
\begin{equation}
    \ket{\psi}_{\hat{k}} := \sum_{l=1}^m z_{l,\hat{k}} \ket{w_l}.
\end{equation}
Hence, for a given basis of $W$ any element $\ket{\psi}\in W$  can be written uniquely as
\begin{equation}
    \ket{\psi} = \sum_{k=1}^{2^{n-1}} \canBasis_k \ket{\psi}_{\hat{k}},\qquad \ket{\psi}_{\hat{k}} \in V.
\end{equation}
From that representation of any element $\ket{\psi}\in W$, we can define the ket projector $P_j:W\rightarrow V$ as
\begin{equation}
    \label{ketproj}
    P_j\ket{\psi} := \ket{\psi}_{\hat{\jmath}},\qquad j=1,\ldots,2^{n-1}.
\end{equation}
\begin{remark}
   Without ambiguity, we use the same notation $P_j$ both for the multicomplex projector $P_j:\mathbb M_n\rightarrow \mathbb C$ defined in (\ref{canonicalrep}) and the module projector $P_j:W\rightarrow V$ defined in (\ref{ketproj}). 
\end{remark}
For any $\ket{\psi}, \ket{\phi} \in W$ and $\alpha\in \mathbb M_n$, one can show that
\begin{equation}
    \begin{aligned}
        P_j(\ket{\psi} +\ket{\phi})&=P_j\ket{\psi}+P_j\ket{\phi}=\ket{\psi}_{\hat\jmath}+\ket{\phi}_{\hat\jmath}, \\*[2ex]
        P_j(\alpha\ket{\psi})&=P_j(\alpha) P_j\ket{\psi}=\alpha_{\hat\jmath} \ket{\psi}_{\hat\jmath},
    \end{aligned}
\end{equation}
 where $P_j : W \rightarrow V$ is the multicomplex projector. 

 \begin{remark}
    The definitions of $V$ in (\ref{defV}) and the module projector $P_j$ in (\ref{ketproj}) depend on the choice of the basis $\{ \ket{w_l} \}_{k=1}^m$ since each $\ket{w_l}$ could be expanded in a new basis with multicomplex coefficients. The elements of $V$ or any projected ket $\ket{\psi}_{\hat{\jmath}}$ would not necessarily have complex coefficients when written in this new basis.
\end{remark}

 We say that a ket $\ket{\psi}$ is in the {\em null cone} of $W$ if for at least one $1\leq j\leq 2^{n-1}$, $\ket{\psi}_{\hat{\jmath}} = 0$. Using that framework, we now turn our attention to the properties of an arbitrary basis of $W$.
\begin{theorem}
    No basis elements of a free $\mathbb{M}_n$-module can belong to the null cone.
\end{theorem}
\begin{proof}
    Let $\ket{w_p}$ be an element of a basis of $W$. We can write 
    \[ \ket{w_p} = \sum_{k=1}^{2^{n-1}} \canBasis_k \ket{w_p}_{\hat{k}}. \]
    If we suppose that $\ket{w_p}$ is in the null cone, then $\ket{w_p}_{\hat{l}} = 0$ for at least one $\hat{l}$. This implies 
    \[ \canBasis_l \ket{w_p} = \canBasis_l \ket{w_p}_{\hat{l}} = 0, \]
    but this last equation contradicts linear independence of the basis.
\end{proof}

For a given basis $\{w_l\}_{l=1}^{m}$ of $W$, let  us now define
    \[ 
\varepsilon_k V:=\big\{ \canBasis_k \sum_{l=1}^m z_l \ket{w_l} \mid z_l \in \mathbb{C} \big\},\qquad k = 1,\ldots,2^{n-1}. 
\]

Trivially we see that any $\varepsilon_k V$ is an $m$-dimensional vector space over $\mathbb{C}$ isomorphic to $V$, where $\{\canBasis_k \ket{w_l}\}_{l=1}^{m}$ is a basis. From (\ref{eq:ketDoubleSum}) we see that the $2^{n-1}\cdot m$ elements $\canBasis_k \ket{w_l}$ are linearly independent over $\mathbb{C}$. Indeed, $\ket{\psi} = 0$ if and only if $z_{l,\hat{k}} = 0$ for all $l,k$, which implies the uniqueness of that representation. Moreover, each term in the summation over $k$ belongs to $\canBasis_k V$. These two statements lead us to the next theorem.
\begin{theorem}
    \label{thm : basesSubspaces : directSum}
    An $m$-dimensional $\mathbb{M}_n$-module $W$ is a $(2^{n-1}\,m)$-dimensional vector space over $\mathbb{C}$ where
    \[ W = \canBasis_1 V \oplus \canBasis_2 V \oplus \cdots \oplus \canBasis_{2^{n-1}} V. \]
\end{theorem}
\begin{theorem}
Let $P_j:W\rightarrow V$ be a projector for a given $1\leq j\leq 2^{n-1}$ and $V$ the subspace defined with respect to the basis $\{ \ket{w_l} \}_{l=1}^m$ of $W$. If $\{ \ket{s_l} \}_{l=1}^m$ is another basis of $W$, then $\{ P_j \ket{s_l} \}_{l=1}^m$ is a basis of $V$.
\end{theorem}
\begin{proof}
    We first show that the kets $P_j\ket{s_1},P_j\ket{s_2},\ldots,P_j\ket{s_m}$ are linearly independent for any fixed value $j$. Let $\alpha_l \in \mathbb{C}$ for $l=1,\ldots,m$ and
    \[ \sum_{l=1}^m \alpha_l P_j\ket{s_l} = 0. \]
    By defining $\beta_l := \alpha_l \canBasis_j$, we obtain $P_k(\beta_l)=\delta_{kj}\alpha_l$. From which we get
    \[ \begin{array}{l}
         P_k\Big( \displaystyle \sum_{l=1}^m \beta_l \ket{s_l} \Big) = \displaystyle \sum_{l=1}^m P_k(\beta_l) P_k\ket{s_l}=\displaystyle \sum_{l=1}^m \delta_{kj} \alpha_l P_k\ket{s_l}=\displaystyle \sum_{l=1}^m \alpha_l P_j\ket{s_l}=0.
    \end{array} \]
    Since this last equation is valid for every $k=1,\ldots,2^{n-1}$, we have $\sum_{l=1}^m \beta_l \ket{s_l}=0$. The set $\{ \ket{s_l} \}_{l=1}^m$ being a basis of $W$, for $l=1,\ldots, m$ we have $\beta_l=0$ such that $\alpha_l = 0$  and $\{ P_j\ket{s_l} \}_{l=1}^m$ is a linearly independent set. 
    
    We now show that this set span $V$. Let $\ket{\phi} \in V$ and consider the ket
    \[ \ket{\psi} = \canBasis_j \ket{\phi} \in W. \]
    Since the multicomplex span of $\{ \ket{s_l} \}_{l=1}^m$ is $W$, there exists $\beta_l \in \mathbb{M}_n$ such that $\sum_{l=1}^m \beta_l \ket{s_l} = \ket{\psi}.$
    Therefore,
    \[ \ket{\phi} = P_j\ket{\psi} = P_j \big( \sum_{l=1}^m \beta_l \ket{s_l} \big) = \sum_{l=1}^m P_j(\beta_l) P_j\ket{s_l}. \]
    Thus the complex span of $\{ P_j\ket{s_l} \}_{l=1}^m$ is $V$.
\end{proof}
\begin{proposition}
Let $\ket{\psi}\in W$ and $1\leq j\leq 2^{n-1}$. If $\ket{\psi}_{\hat{\jmath}}=P_j\ket{\psi}=0$ for a given basis of $W$, then the $j$-th projection of $\ket{\psi}$ is zero for any other basis of $W$.
\end{proposition}
\begin{proof}
For a given basis $\{\ket{w_l}\}_{l=1}^m$ of $W$, $\ket{\psi}$ can be written as $\ket{\psi}=\sum_{k=1}^{2^{n-1}}\canBasis_k \ket{\psi}_{\hat k}$, where $\ket{\psi}_{\hat k}=\sum_{l=1}^m z_{l,\hat k}\ket{w_l}$ for $z_{l,\hat k}\in \mathbb C$. If $\ket{\psi}_{\hat\jmath}=P_j\ket{\psi}=0$ in this basis then $z_{1,\hat{\jmath}}=\cdots =z_{m,\hat{\jmath}}=0$. Suppose now a second basis $\{\ket{\widetilde w_p}\}_{p=1}^m$ of $W$. In a similar way we find $\ket{\psi}=\sum_{k=1}^{2^{n-1}}\canBasis_k \ket{\widetilde\psi}_{\hat k}$, where $\ket{\widetilde\psi}_{\hat k}=\sum_{l=1}^m \widetilde z_{l,\hat k}\ket{\widetilde w_l}$ for $\widetilde z_{l,\hat k}\in \mathbb C$. We want to show that
\begin{equation}
\label{matrix2bases}
\ket{\widetilde\psi}_{\hat\jmath}=\sum_{l=1}^m \widetilde z_{l,\hat \jmath}\ket{\widetilde w_l}=0 \qquad \Leftrightarrow \qquad z_{1,\hat{\jmath}}=\cdots =z_{m,\hat{\jmath}}=0.
\end{equation}

The two bases are related by a  multicomplex nonsingular matrix $(\gamma_{lp})$ \cite{FiniteHilbert} such that
$$
\ket{w_l}=\sum_{p=1}^m \gamma_{lp} \ket{\widetilde w_p},\qquad \gamma_{lp}\in \mathbb M_n.
$$
and $\ket{\psi}_{\hat k}$ can be expressed as
$$
\ket{\psi}_{\hat k}=\sum_{l=1}^m z_{l,\hat k}\sum_{p=1}^m \gamma_{lp}\ket{\widetilde w_p}=\sum_{p=1}^m\Big(\sum_{l=1}^m z_{l,\hat k}\gamma_{lp}\Big)\ket{\widetilde w_p}.
$$
Therefore, we have $\widetilde z_{p,\hat k}=\sum_{l=1}^m z_{l,\hat k}\gamma_{lp}$ and, in particular, we find that (\ref{matrix2bases}) is satisfied.
\end{proof}

\subsection{Multicomplex matrices and determinants}

A multicomplex $m \times m$ square matrix $A$ is an array of $m^2$ multicomplex numbers $A_{ij}$. Each element can be written as
\begin{equation}
    A_{ij} = \sum_{k=1}^{2^{n-1}} a_{ij\hat{k}} \canBasis_k,\qquad a_{ij\hat{k}} \in \mathbb{C}.
\end{equation}
The matrix itself is then
\begin{equation}
    A = (A_{ij})_{m \times m} = \Big( \sum_{k=1}^{2^{n-1}} a_{ij\hat{k}} \canBasis_k \Big)_{m \times m} = \sum_{k=1}^{2^{n-1}} \canBasis_k (a_{ij\hat{k}})_{m \times m}.
\end{equation}
Applying the projection operator $P_{l}$ on $A$ gives the $m^2$ associated complex matrix:
\begin{equation}
    A_{\hat{l}} = P_l A = (a_{ij\hat{l}})_{m \times m}.
\end{equation}
Thus any multicomplex square matrix $A$ can be written as
\begin{equation}
    A = \sum_{k=1}^{2^{n-1}} \canBasis_k A_{\hat{k}}, \qquad A_{\hat{k}} \in M_{m \times m}(\mathbb{C}).
\end{equation}
\begin{theorem}
    Let $A = \displaystyle\sum_k \canBasis_k A_{\hat{k}}$ be an $m \times m$ multicomplex matrix. Then
    \[ \det A = \sum_{k=1}^{2^{n-1}} \canBasis_k \det A_{\hat{k}}. \]
\end{theorem}
\begin{proof}
    Let $\{C_i\}_{i=1}^m$ be the set of columns of $A$, so that $A = (C_1, C_2, \ldots, C_m)$. We can write the $i$th column 
    \[ C_i = \sum_{k=1}^{2^{n-1}} \canBasis_k C_i^{(k)}, \]
    where the columns $C_i^{(k)}$ are complex.
    Since the determinant is a multilinear function, for $C_1$ we have
    \[
        \det A=\det \Big(\displaystyle \sum_{k=1}^{2^{n-1}} \canBasis_k C_1^{(k)},C_2, \ldots, C_m \Big)= \displaystyle \sum_{k=1}^{2^{n-1}} \canBasis_k\det\Big(C_1^{(k)},C_2,\ldots, C_m\Big).
     \]
    By applying thi procedure successively for the remaining columns $C_2,\ldots , C_m$ we find
    
    \[ \det A= \det(C_1,\ldots,C_m) = \sum_{k=1}^{2^{n-1}} \canBasis_k \det(C_1^{(k)},\ldots,C_m^{(k)}). \]
\end{proof}

From the previous theorem we can see that $\det A = 0$ if and only if $\det A_{\hat{l}} = 0$ for  $l=1,\ldots,m$. Moreover, $\det A$ is in the null cone if $\det A_{\hat{l}} = 0$ for at least one $l$.
\begin{definition}
    A multicomplex square matrix is singular if its determinant is in the null cone.
\end{definition}
\begin{theorem}
    \label{thm : inverseSingular}
    The inverse $A^{-1}$ of a multicomplex square matrix $A$ exists if and only if $A$ is nonsingular. Then $A^{-1}$ is given by $A^{-1} = \displaystyle \sum_{k=1}^{2^{n-1}} \canBasis_k (A_{\hat{k}})^{-1}$. 
\end{theorem}
\begin{proof}
    If $A^{-1}$ exists then $A^{-1}A = I$ such that $1 = \det(A^{-1}A) = \det A^{-1} \det A$.
    Consequently, $\det A$ is not in the null cone. Conversely, if $A$ is nonsingular then $\det A_{\hat{l}} \neq 0$ for $l=1,\ldots, 2^{n-1}$ and
    \[ \Big( \sum_{k=1}^{2^{n-1}} \canBasis_k (A_{\hat{k}})^{-1} \Big) \cdot \Big( \sum_{k=1}^{2^{n-1}} \canBasis_k A_{\hat{k}} \Big) = \sum_{k=1}^{2^{n-1}} \canBasis_k I = I. \]
\end{proof}

\subsection{Linear operators}

\begin{definition}
    A multicomplex linear operator is a function $A : W \rightarrow W$ such that for all $\alpha, \beta \in \mathbb{M}_n$ and $\ket{\psi}, \ket{\phi} \in W$,
    \[ A(\alpha \ket{\psi} + \beta \ket{\phi}) = \alpha A\ket{\psi} + \beta A\ket{\phi}. \]
\end{definition}
For any $\ket{\psi} \in W$, we set $A_{\hat{\jmath}} \ket{\psi}:= P_jA\ket{\psi}$, then we obtain
\begin{equation}
    \label{eq:PjA}
    \begin{array}{rcl}
    A_{\hat{\jmath}} \ket{\psi} &=&P_jA\ket{\psi}=P_jA \displaystyle \sum_{k=1}^{2^{n-1}} \canBasis_k \ket{\psi}_{\hat{k}} =P_j\Big( \sum_{k=1}^{2^{n-1}} \canBasis_k A\ket{\psi}_{\hat{k}} \Big) \\*[2ex]
    &=& \displaystyle \sum_{k=1}^{2^{n-1}} P_j(\canBasis_k) P_jA \ket{\psi}_{\hat{k}}  = P_j A \ket{\psi}_{\hat{\jmath}} = A_{\hat{\jmath}} \ket{\psi}_{\hat{\jmath}}.
    \end{array}
\end{equation}
\begin{remark}
Note that for an arbitrary ket $\ket{\psi}\in W$, we have $AP_j\ket{\psi} \in W$ and $P_jA\ket{\psi}\in V$, hence a multicomplex linear operator $A$ do not commute with $P_j$ in general. 
\end{remark}
 \begin{theorem}
     A multicomplex linear operator $A:W\rightarrow W$ commutes with the ket projector $P_j:W\rightarrow V$ for $1\leq j\leq 2^{n-1}$ if and only if $AW\subseteq V$.
 \end{theorem}
 \begin{proof}
     Suppose $A$ and $P_j$ commute, then for all $\ket{\psi}\in W$ we have 
     $$
    0=(P_jA-AP_j)\ket{\psi}=A_{\hat\jmath}\ket{\psi}_{\hat\jmath}-A\ket{\psi}_{\hat \jmath}
     $$
     such that $A=A_{\hat \jmath}$ and $A:W\rightarrow V$.
     
     Conversely if we suppose $AW\subseteq V$, for any $\ket{\psi}\in W$ we have 
    $$
    P_jA\ket{\psi}=P_jA\ket{\psi}_{\hat \jmath}=A\ket{\psi}_{\hat \jmath}
    $$
    and
    $$
    AP_j\ket{\psi}=A\ket{\psi}_{\hat \jmath}
    $$
    such that $[A,P_j]=0$.
 \end{proof}
\begin{definition}
    A multicomplex linear operator $A:W\rightarrow W$ belongs to the null cone if for at least one $1\leq j\leq 2^{n-1}$ we have $A_{\hat{\jmath}} = 0$.
\end{definition}
\begin{definition}
    Let $A:W\rightarrow W$ be a multicomplex linear operator and let
    \[ 
    A \ket{\psi} = \lambda \ket{\psi}, \qquad \lambda \in \mathbb{M}_n
    \]
    where $\ket{\psi}$ is not in the null cone. Then $\lambda$ is called an eigenvalue of $A$ and $\ket{\psi}$ is the corresponding eigenket.
\end{definition}
Expanding the expression $A \ket{\psi} = \lambda \ket{\psi}$ we get
\begin{equation}
    A \sum_{k=1}^{2^{n-1}} \canBasis_k \ket{\psi}_{\hat{k}} = \lambda \sum_{k=1}^{2^{n-1}} \canBasis_k \ket{\psi}_{\hat{k}} \quad \Rightarrow \quad \sum_{k=1}^{2^{n-1}} \canBasis_k A \ket{\psi}_{\hat{k}} = \sum_{k=1}^{2^{n-1}} \canBasis_k \lambda \ket{\psi}_{\hat{k}}.
\end{equation}
Applying $P_j$ on both sides gives us
\begin{equation}
    \label{eq : linearOp : projectedEigenval}
    A_{\hat{\jmath}} \ket{\psi}_{\hat{\jmath}} = \lambda_{\hat{\jmath}} \ket{\psi}_{\hat{\jmath}},\qquad \lambda_{\hat{\jmath}} = P_j(\lambda).
\end{equation}
This last equation means that the eigenvalue of a ket projection of $A$ is the corresponding multicomplex projection of $\lambda$.
\begin{theorem}
    Let $A$ and $B$ be two multicomplex linear operators. Then for all $1\leq j\leq 2^{n-1}$ we have the following properties:
    \[ 
         P_j(A+B) = A_{\hat{\jmath}} + B_{\hat{\jmath}} \qquad \text{and} \qquad   P_j(AB) = A_{\hat{\jmath}} B_{\hat{\jmath}}.
    \]
\end{theorem}
\begin{proof}
    Let $\ket{\psi} \in W$, then 
    \[ \begin{array}{r c l}
        P_j(A+B)\ket{\psi} & = & P_j(A \ket{\psi} + B \ket{\psi}) =A_{\hat{\jmath}} \ket{\psi}_{\hat{\jmath}} + B_{\hat{\jmath}} \ket{\psi}_{\hat{\jmath}} = A_{\hat{\jmath}} \ket{\psi} + B_{\hat{\jmath}} \ket{\psi}\\*[2ex]
        & = & (A_{\hat{\jmath}} + B_{\hat{\jmath}})\ket{\psi}.\\
    \end{array} \]
    Writing $\ket{\phi} := B \ket{\psi}$, we have
    \[ 
    P_j A B \ket{\psi}  =  P_j A \ket{\phi} = A_{\hat{\jmath}} \ket{\phi}_{\hat{\jmath}}=  A_{\hat{\jmath}} P_j \ket{\phi} = A_{\hat{\jmath}} P_j B \ket{\psi}=  A_{\hat{\jmath}}B_{\hat{\jmath}}\ket{\psi}.
    \]
\end{proof}
Once again, an equivalent proof of the next theorem is found in \cite{FiniteHilbert}, since the case with bicomplex numbers is easily generalized to that of multicomplex numbers from the previous results.
\begin{theorem}
    The action of a linear multicomplex operator on $W$ can be represented by a multicomplex matrix.
\end{theorem}

\section{Multicomplex Hilbert spaces}

\subsection{Scalar product}

\begin{definition}
    \label{def : hilbertSpaces : defScalarProduct}
    The multicomplex scalar product is a function associating a multicomplex number to each pair of elements $\ket{\psi}, \ket{\phi} \in W$ which for all $\ket{\chi} \in W$ and $\alpha \in \mathbb{M}_n$ satisfies:
    \begin{enumerate}
        \item $(\ket{\psi}, \ket{\phi} + \ket{\chi}) = (\ket{\psi}, \ket{\phi}) + (\ket{\psi}, \ket{\chi})$;
        \item $(\ket{\psi}, \alpha \ket{\phi}) = \alpha (\ket{\psi}, \ket{\phi})$;
        \item $(\ket{\psi}, \ket{\phi}) = (\ket{\phi}, \ket{\psi})^{\Lambda}$;
        \item $(\ket{\psi}, \ket{\psi}) \in \mathbb{D}_n^+$ and $(\ket{\psi}, \ket{\psi})=0$ if and only if $\ket{\psi} = 0$.
    \end{enumerate}
\end{definition}
Here $\mathbb{D}_n^+$ is the set of multicomplex numbers with real and positive components in the idempotent canonical representation. The $j$th projection of the multicomplex scalar product is denoted
\begin{equation}
    (\cdot\,,\,\cdot)_{\hat{\jmath}} := P_j(\cdot\,,\,\cdot)
\end{equation}
and from Definition \ref{def : hilbertSpaces : defScalarProduct} it follows that any projection is itself a well defined standard scalar product on the associated vector space $\canBasis_j V$ (as well as on $V$), i.e.
\begin{equation}
\label{scprodV}
\begin{array}{l}
(\ket{\psi}, \ket{\phi} + \ket{\chi})_{\hat\jmath} = (\ket{\psi}, \ket{\phi})_{\hat\jmath} + (\ket{\psi}, \ket{\chi})_{\hat\jmath},\quad (\ket{\psi}, \alpha \ket{\phi})_{\hat\jmath} = \alpha_{\hat\jmath} (\ket{\psi}, \ket{\phi})_{\hat\jmath},\\*[2ex]
(\ket{\psi}, \ket{\phi})_{\hat\jmath} = \overline{(\ket{\phi}, \ket{\psi})}_{\hat\jmath} \quad \text{and}\quad
(\ket{\psi}, \ket{\psi})_{\hat \jmath} \in \mathbb{R}^+,\quad (\ket{\psi}, \ket{\psi})_{\hat \jmath}=0\ \Leftrightarrow \ket{\psi}_{\hat \jmath}= 0.
\end{array}
\end{equation}
\begin{proposition}
    A ket $\ket{\psi}\in W$ is in the null cone if and only if its multicomplex scalar product $(\ket{\psi}, \ket{\psi})$ is in the null cone.
\end{proposition}
\begin{proof}
    It's a direct consequence of the last equivalence in (\ref{scprodV}).
\end{proof}
\begin{theorem}
    \label{thm : hilbertSpaces : decompositionOfScalarProd}
    Let $\ket{\psi}, \ket{\phi} \in W$, then
    \[ (\ket{\psi}, \ket{\phi}) = \sum_{k=1}^{2^{n-1}} (\ket{\psi}_{\hat{k}}, \ket{\phi}_{\hat{k}})_{\hat{k}}\,\canBasis_k. \]
\end{theorem}
\begin{proof}
    \[ \begin{array}{r c l}
         (\ket{\psi}, \ket{\phi}) & = & \big( \displaystyle \sum_{k=1}^{2^{n-1}} \ket{\psi}_{\hat{k}} \canBasis_k, \displaystyle \sum_{l=1}^{2^{n-1}} \ket{\phi}_{\hat{l}} \canBasis_l \big)= \displaystyle \sum_{k=1}^{2^{n-1}} \displaystyle \sum_{l=1}^{2^{n-1}} (\ket{\psi}_{\hat{k}}, \ket{\phi}_{\hat{l}})\, \canBasis_k^{\Lambda}\ \canBasis_l\\
         & = & \displaystyle \sum_{k=1}^{2^{n-1}} (\ket{\psi}_{\hat{k}}, \ket{\phi}_{\hat{k}}) \canBasis_k=  \displaystyle \sum_{k=1}^{2^{n-1}} \Big( \displaystyle \sum_{j=1}^{2^{n-1}} (\ket{\psi}_{\hat{k}}, \ket{\phi}_{\hat{k}})_{\hat{\jmath}} \,\canBasis_j \Big) \canBasis_k\\
         & = & \displaystyle \sum_{k=1}^{2^{n-1}} (\ket{\psi}_{\hat{k}}, \ket{\phi}_{\hat{k}})_{\hat{k}} \,\canBasis_k.
    \end{array} \]
\end{proof}
We know that $V$ and $\canBasis_j V$ for all $j = 1,\ldots,2^{n-1}$ are finite-dimensional vector spaces on the complex $\mathbb{C}$. Furthermore, for each $j$, $(\cdot\,,\, \cdot)_{\hat{\jmath}}$ is a standard complex scalar product on both $V$ and $\canBasis_j V$, which means that each of these spaces equipped with the projected scalar product is a 
finite-dimensional Hilbert space. From Theorem \ref{thm : basesSubspaces : directSum}, $W$ is a direct sum of the spaces $\canBasis_j V$ and is thus itself a standard Hilbert space equipped with the following scalar product :
\begin{equation}
    \label{eq : scalProd : scalProdComplexValued}
    \Big( \sum_{k=1}^{2^{n-1}} \ket{\psi_k}, \sum_{l=1}^{2^{n-1}} \ket{\phi_l} \Big)_{\mathbb{C}} := \sum_{k=1}^{2^{n-1}} (\ket{\psi_k}, \ket{\phi_k})_{\hat{k}},\qquad \ket{\psi_k},\ket{\phi_k}\in \canBasis_k V.
\end{equation}
    This function induce a norm and a metric from which $W$ is a complete metric space. There is an important distinction to make between the scalar products defined on $W$ in Definition~\ref{def : hilbertSpaces : defScalarProduct} and in equation (\ref{eq : scalProd : scalProdComplexValued}): one is multicomplex-valued and the other is complex-valued respectively. Moreover, from expression (\ref{eq : scalProd : scalProdComplexValued}) the complex-valued scalar product is clearly induced from the multicomplex one, but the converse is also possible. Indeed, take $(\cdot\,,\,\cdot)_{\mathbb{C}}$ as an independently defined complex-valued scalar product on $W$ (considered as a vector space over $\mathbb{C}$), then
\begin{equation}
    ( \ket{\psi}, \ket{\phi} ) := \sum_{k=1}^{2^{n-1}} (\ket{\psi}_{\hat{k}}, \ket{\phi}_{\hat{k}})_{\mathbb{C}}\,\canBasis_k
\end{equation}
is a multicomplex scalar product on $W$.
\begin{theorem}
    Any $\mathbb{M}_n$-module $W$ is a Hilbert space if and only if $W$ is equipped with a multicomplex scalar product.
\end{theorem}
The $\mathbb{M}_n$-module is consequently a special case of a module in which we can coherently mix these two notions together and say that from the existence of a multicomplex scalar product, $W$ is a \textbf{multicomplex Hilbert space}.
\begin{theorem}[Riesz]
    \label{thm : scalarProd : Riesz}
    Let $f : W \rightarrow \mathbb{M}_n$ be a linear functional on $W$. Then there exists a unique $\ket{\psi} \in W$ such that $\forall \ket{\phi} \in W$, $f(\ket{\phi}) = (\ket{\psi}, \ket{\phi})$.
\end{theorem}
\begin{proof}
    Any projection $f_{\hat{\jmath}}$ of $f$ is a linear functional on $V$. Applying the classical Riesz theorem, there exists a unique $\ket{\psi_j} \in V$ such that for all $\ket{\phi_j} \in V$, $f_{\hat{j}} (\ket{\phi_j}) = (\ket{\psi_j}, \ket{\phi_j})_{\hat{j}}$. We set $\ket{\psi} := \sum_{k=1}^{2^{n-1}} \ket{\psi_k} \canBasis_k$ and use theorem \ref{thm : hilbertSpaces : decompositionOfScalarProd} to get
    \[ (\ket{\psi}, \ket{\phi}) = \sum_{k=1}^{2^{n-1}} (\ket{\psi_k}, \ket{\phi}_{\hat{k}})_{\hat{k}} \canBasis_k = \sum_{k=1}^{2^{n-1}} f_{\hat{k}} (\ket{\phi}_{\hat{k}}) \canBasis_k = f(\ket{\phi}). \]
\end{proof}
From this generalization of Riezs theorem, linear functionals are in one-to-one correspondence with kets and can be replaced by the scalar product operation. This allows the use of Dirac's notation and the alternative writing of the scalar product as $\braket{\psi}{\phi} := (\ket{\psi}, \ket{\phi})$.
\begin{theorem}
    Any ket $\ket{\psi} \in W$ not in the null cone can be normalized.
\end{theorem}
\begin{proof}
    For any $\ket{\psi} \in W$ not in the null cone, we have $(\ket{\psi}, \ket{\psi}) \in \mathbb{D}_n^+$ and this scalar product has strictly real positive components:
    \[ (\ket{\psi}, \ket{\psi}) = \sum_{k=1}^{2^{n-1}} a_k \,\canBasis_k, \qquad a_k > 0. \]
    Then the ket
    \[ \ket{\phi} := \big( \sum_{k=1}^{2^{n-1}} \frac{1}{\sqrt{a_k}} \canBasis_k \big) \ket{\psi} \]
    satisfies $(\ket{\phi}, \ket{\phi}) = 1$.
\end{proof}

\subsection{Spectral decomposition theorem}

\begin{definition}
    Let $A$ be a linear operator on $W$. Then the adjoint operator $A^*$ of $A$ is defined as an operator on $W$ satisfying the following equality.
    \[ (\ket{\psi}, A \ket{\phi}) := (A^* \ket{\psi}, \ket{\phi}), \qquad \forall \ket{\psi}, \ket{\phi} \in W. \]
\end{definition}
By the decomposition of $A$ in its components $A_{\hat{k}}$ (linear operators on $V$) under the canonical idempotent representation, the adjoint always exists, is unique, and his expression is given by the linear operator for which each component is the adjoint $A_{\hat{k}}^*$ of $A_{\hat{k}}$, i.e. for 
$$
A = \sum_{k=1}^{2^{n-1}} A_{\hat{k}} \canBasis_k \qquad \text{then} \qquad  A^* = \sum_{k=1}^{2^{n-1}} A_{\hat{k}}^* \canBasis_k.
$$
Thus the multicomplex operator satisfies the same basic properties of the adjoint on an usual complex vector space. Moreover,
\begin{equation}
    [P_j(A)]^* = P_j(A^*) \qquad \text{for} \qquad j=1,\ldots,2^{n-1}.
\end{equation}
Now let $\ket{\psi}, \ket{\phi} \in W$. We define the operator $\ket{\phi} \bra{\psi}$ so that its action on an arbitrary ket $\ket{\chi} \in W$ is given by
\begin{equation}
    (\ket{\phi} \bra{\psi}) \ket{\chi} := \ket{\phi} (\braket{\psi}{\chi}).
\end{equation}
From the generalized Riesz theorem, the action of $\bra{\psi}$ on a ket is a linear functional and always gives a scalar, which means that the operator $\ket{\phi} \bra{\psi}$ itself is linear. 

As in a standard Hilbert space, the identity operator can be written in terms of any orthonormal basis $\{ \ket{u_l} \}_{l=1}^m$ of $W$ :
\begin{equation}
    \sum_{l=1}^m \ket{u_l} \bra{u_l} = \mathrm{id}.
\end{equation}
Indeed, the actions of the left-hand side on any basis element $\ket{u_p}$ is given by
\begin{equation}
    \big( \sum_{l=1}^m \ket{u_l} \bra{u_l} \big) \ket{u_p} = \sum_{l=1}^m \ket{u_l} \braket{u_l}{u_p} = \sum_{l=1}^m \delta_{lp} \ket{u_l} = \ket{u_p}.
\end{equation}
\begin{definition}
    A multicomplex linear operator $A$ is called self-adjoint if $A^* = A$.
\end{definition}
Supposing that $A$ is a multicomplex self-adjoint operator, then
\begin{equation}
    A = A^* \Leftrightarrow \sum_{k=1}^{2^{n-1}} A_{\hat{k}} \canBasis_k = \sum_{k=1}^{2^{n-1}} A_{\hat{k}}^* \canBasis_k \Leftrightarrow A_{\hat{k}} = A_{\hat{k}}^* \quad \text{for}\quad k = 1,\ldots,2^{n-1}
\end{equation}
and for any projection $[P_j(A)]^* = P_j(A^*) = P_j(A)$. It implies that the projection of a self-adjoint operator on $W$ is itself a self-adjoint operator on $V$.
\begin{theorem}
    Let $A$ be a self-adjoint operator on $W$. Then the eigenvalues of $A$ associated to an eigenket not in the null cone are all in the set $\mathbb{D}_n$ of multicomplex numbers with real components.
\end{theorem}
\begin{proof}
    Let $\ket{\psi}$ be an eigenket of $A$ not in the null cone. From equation (\ref{eq : linearOp : projectedEigenval}),
    \[ A \ket{\psi} = \lambda \ket{\psi} \quad \Rightarrow\quad  A_{\hat{\jmath}} \ket{\psi}_{\hat{\jmath}} = \lambda_{\hat{\jmath}} \ket{\psi}_{\hat{\jmath}}, \qquad \ket{\psi}_{\hat{j}} \neq 0, \qquad j = 1,\ldots,2^{n-1}. \]
    If $A$ is self-adjoint then its projections are also self-adjoint and this implies that $\lambda_{\hat{\jmath}}$ is a real number for $j = 1,\ldots,2^{n-1}$.
\end{proof}
\begin{theorem}
    Two eigenkets of a multicomplex self-adjoint operator are orthogonal if the difference of the two associated eigenvalues is not in the null cone.
\end{theorem}
\begin{proof}
    Let $\ket{\psi}$ and $\ket{\phi}$ be two eigenkets of a self-adjoint operator $A$ on $W$ with associated eigenvalues $\lambda$ and $\lambda'$ respectively. Then
    \[ \begin{aligned}
        0 & = (\ket{\psi}, A \ket{\phi}) - (\ket{\phi}, A \ket{\psi})^{\Lambda}\\
        & = \lambda' (\ket{\psi}, \ket{\phi}) - \lambda^{\Lambda} (\ket{\phi}, \ket{\psi})^{\Lambda}\\
        & = (\lambda' - \lambda^{\Lambda}) (\ket{\psi}, \ket{\phi}).
    \end{aligned} \]
    Since $\lambda \in \mathbb{D}_n$, $\lambda^{\Lambda} = \lambda$ and if $\lambda' - \lambda$ is not in the null cone then $(\ket{\psi}, \ket{\phi}) = 0$.
\end{proof}
\begin{theorem}[Spectral decomposition] Let $W$ be a finite-dimensional free $\mathbb{M}_n$-module and let $A : W \rightarrow W$ be a multicomplex self-adjoint operator. It is always possible to find a set $\{ \ket{\psi_l} \}_{l=1}^m$ of eigenkets of $A$ that makes up an orthonormal basis of $W$. Moreover, $A$ can be expressed as
    \[ 
    A = \sum_{l=1}^m \lambda_l \ket{\psi_l} \bra{\psi_l},
    \]
    where $\lambda_l$ is the eigenvalue of $A$ associated with the eigenket $\ket{\psi_l}$.
\end{theorem}
\begin{proof}
    Any projection $P_j(A) = A_{\hat{\jmath}}$ is a self-adjoint linear operator on $V$. Applying the standard spectral decomposition theorem, for all $j = 1,\ldots,2^{n-1}$ there exists an orthonormal set $\{ \ket{\psi_l}_{\hat{\jmath}} \}_{l=1}^m$ of eigenkets of $A_{\hat{\jmath}}$ which is also a basis of $V$ with respect to the scalar product $(\cdot\,,\, \cdot)_{\hat{j}}$. For 
    \[ \ket{\psi_l} := \sum_{k=1}^{2^{n-1}} \ket{\psi_l}_{\hat{k}}\, \canBasis_k \]
    then the set $\{ \ket{\psi_l} \}_{l=1}^m$ satisfies the statement of this theorem and
    \[ \begin{aligned}
        A & = \sum_{k=1}^{2^{n-1}} A_{\hat{k}} \,\canBasis_k = \sum_{k=1}^{2^{n-1}} \Big( \sum_{l=1}^m \lambda_{l,\hat{k}} \ket{\psi_l}_{\hat{k}} \bra{\psi_l}_{\hat{k}} \Big) \canBasis_k\\
        & = \sum_{l=1}^m \sum_{k=1}^{2^{n-1}} P_k\big(\lambda_l \ket{\psi_l} \bra{\psi_l}\big) \canBasis_k = \sum_{l=1}^m \lambda_l \ket{\psi_l} \bra{\psi_l}
    \end{aligned} \]
    where $\lambda_{l,\hat{k}}$ represents the complex eigenvalue associated with the eigenket $\ket{\psi_l}_{\hat{k}}$.
\end{proof}

\section{Conclusion}

The multicomplex idempotent canonical basis seems to be the first step to a full understanding of the multicomplex number space. It is a natural basis to represent the principal ideals of this structure, but also a way to greatly simplify the algebra as both addition and multiplication become componentwise. The composition of conjugates $\Lambda_n = \dagger_1 \cdots \dagger_n$ has been introduced, it is an involution on the multicomplex numbers which plays a significant role in characterizing multiperplex numbers and also naturally defines a multiperplex-valued norm on $\mathbb{M}_n$. The multicomplex version of modules and Hilbert spaces share the same general properties of their complex counterparts, with some differences or more specific cases due to the presence of zero divisors. Following this work, we expect several new results where multicomplex algebra can be applied in fundamental or applied mathematics as well as in physics.

\bibliographystyle{plain}
\bibliography{refs}

\vspace*{1cm}

\noindent $^1$D\'EPARTEMENT DE MATHÉMATIQUES ET D'INFORMATIQUE, \\ UNIVERSITÉ DU QUÉBEC, TROIS-RIVI\`ERES, QC, CANADA \\
{\em Email address}: \texttt{derek.courchesne@uqtr.ca} \\
\\
\noindent $^2$D\'EPARTEMENT DE MATHÉMATIQUES ET D'INFORMATIQUE, \\ UNIVERSITÉ DU QUÉBEC, TROIS-RIVI\`ERES, QC, CANADA \\
{\em Email address}: \texttt{sebastien.tremblay@uqtr.ca}\\

\end{document}